\documentclass[twocolumn,pra,aps,amsfonts,superscriptaddress,nofootinbib]{revtex4}
\usepackage{xcolor}
\usepackage{amsmath}
\usepackage{amssymb}
\usepackage{amsfonts}
\usepackage{amsthm}
\usepackage{dsfont}
\usepackage{physics}
\usepackage{mathpazo}
\usepackage{comment}
\usepackage{mathtools}
\usepackage[colorlinks=true,citecolor= blue,urlcolor=blue]{hyperref}
\usepackage{cancel}
\newtheorem{theorem}{Theorem}[section]
\newtheorem{corollary}{Corollary}[theorem]
\newtheorem{lemma}[theorem]{Lemma}
\begin{document}

\newcommand{\mean}[1]{\left\langle #1\right\rangle}

\newcommand{\tb}[1]{\textcolor{blue}{#1}}
\newcommand{\tgr}[1]{\textcolor{red}{#1}}
\newcommand{\tgrn}[1]{\textcolor{green}{#1}}

\title{An elegant proof of self-testing for multipartite Bell inequalities}
\author{Ekta Panwar}\email{ekta.panwar@phdstud.ug.edu.pl}\affiliation{Institute of Theoretical Physics and Astrophysics, Faculty of Mathematics, Physics and Informatics, University of Gda\'nsk, ul. Wita Stwosza 57, 80-308 Gda\'nsk, Poland}
\author{Palash Pandya}\affiliation{Institute of Theoretical Physics and Astrophysics, Faculty of Mathematics, Physics and Informatics, University of Gda\'nsk, ul. Wita Stwosza 57, 80-308 Gda\'nsk, Poland}
\author{Marcin Wie\'sniak}\affiliation{Institute of Theoretical Physics and Astrophysics, Faculty of Mathematics, Physics and Informatics, University of Gda\'nsk, ul. Wita Stwosza 57, 80-308 Gda\'nsk, Poland}
\affiliation{International Centre for Theory of Quantum Technologies, University of Gda\'nsk, ul. Wita Stwosza 63, 80-308 Gda\'nsk, Poland}

\date{\today}
\begin{abstract}
The predictions of quantum theory are incompatible with local-causal explanations. This phenomenon is called Bell non-locality and is witnessed by violation of Bell-inequalities. The maximal violation of certain 
Bell-inequalities can only be attained in an essentially unique manner. This feature is referred to as \emph{self-testing} and constitutes the most accurate form of certification of quantum devices. While self-testing in bipartite Bell scenarios has been thoroughly studied, self-testing in the more complex multipartite Bell scenarios remains largely unexplored. This work presents a simple and broadly applicable self-testing argument for $N$-partite correlation Bell inequalities with two binary outcome observables per party. Our proof technique forms a generalization of the Mayer-Yao formulation and is not restricted to linear Bell-inequalities, unlike the usual \emph{sum of squares} method. To showcase the versatility of our proof technique, we obtain self-testing statements for $N$ party Mermin-Ardehali-Belinskii-Klyshko (MABK) and Werner-Wolf-Weinfurter-\.Zukowski-Brukner (WWW\.ZB) family of linear Bell inequalities, and Uffink's family of $N$ party quadratic Bell-inequalities. 
  
\end{abstract}
\maketitle
\section{Introduction}

One of the most striking features of quantum theory is the deviation of its predictions for Bell experiments from the predictions of classical theories with local casual (hidden variable) explanations \cite{PhysicsPhysiqueFizika.1.195}. This phenomenon is captured by quantum violation of statistical inequalities, which are satisfied by all local realistic theories, referred to as Bell-inequalities \cite{RevModPhys.86.419}.
Experimental demonstrations of the loophole-free (e.g., \cite{Hensen2015, PhysRevLett.115.250401, PhysRevLett.115.250402}) violation of such inequalities implies the possibility of sharing intrinsically random private numbers among an arbitrary number of spatially separated parties which power unconditionally secure private key distribution schemes (for more information on Device-independent quantum cryptography see \cite{2014, PhysRevLett.97.120405,colbeck2011quantum, Ekert2014, PhysRevLett.67.661}). Such applications of Bell non-locality follow from the fact that the extent of Bell inequality violation can uniquely identify the specific entangled quantum states and measurements, a phenomenon referred to as \textit{self-testing} (see \cite{743501,mayers2004self} for initial contributions and the recent review of the progress till now see \cite{Supic2020selftestingof}). 

Self-testing statements are the most accurate form of certifications for quantum systems.
Self-testing schemes allow us to infer the underlying physics of a quantum experiment, i.e., the state and the measurements (up to local isometry) without any characterization of the internal workings of the measurement devices, and based only on the observed statistics, i.e., treating the measurement devices as black boxes with classical inputs and outputs. Self-testing has found many application in several areas like \textit{device-independent randomness generation} \cite{Colbeck_2011,colbeck2011quantum},\textit{ quantum cryptography} \cite{Ekert2014}, \textit{entanglement detection} \cite{Gheorghiu2019},\textit{ delegated quantum computing } \cite{PhysRevA.98.042336,PhysRevLett.121.180505}. While self-testing in the bipartite Bell scenarios has been thoroughly studied, self-testing in more complex multipartite Bell scenarios (see FIG. \ref{manipuraka}) remains largely unexplored. 
 
In a multipartite setting, self-testing has been demonstrated for graph states using stabilizer operators \cite{10.1007/978-3-642-54429-3_7}. Self-testing of multipartite graph states and partially entangled \emph{Greenberger-Horne-Zeilinger} (\textbf{GHZ})  states has been demonstrated using a stabilizer-based approach, and Bell inequalities explicitly constructed for the state \cite{PhysRevLett.124.020402, greenberger2007going}. In general, multipartite Bell inequalities explicitly tailored for self-testing of a multipartite entangled state can be obtained using linear programming  \cite{PhysRevA.90.042340,PhysRevA.91.022115,PhysRevLett.113.040401,PhysRevLett.121.180505}. However, the self-testing statements must be verified using the numerical technique (such as the Swap method), which tends to be computationally expensive for multipartite scenarios with more than four parties. Completely analytical self-testing statements have also been obtained for multipartite states such as the \textbf{W} and the Dicke states by reprocessing self-testing protocols of bipartite states \cite{PhysRevA.90.042339,fadel2017selftesting,_upi__2018}. Furthermore, parallel self-testing statements for multipartite states can be obtained using categorical quantum mechanics \cite{PhysRevA.72.052103}. Finally, the Mayers-Yao criterion can also be utilized for self-testing of graph states when the underlying graph is a triangular lattice \cite{PhysRevA.97.052308}. 

Mayers and Yao's criterion states that if given a Bell inequality is maximally violated, the shared quantum state must be equivalent to a reference state up to local isometries, and the optimal observables are (almost) uniquely defined with respect to the state. In this article, we prove Mayer-Yao-like self-testing statements for multipartite Bell scenarios without relying on the Bell-operator dependent sum-of-squares decomposition \cite{PhysRevA.91.052111}. Consequently, our methodology immediately extends to all multipartite Bell scenarios, where each spatially separated party has two observables. To exemplify our proof technique, we obtain self-testing statements for $N$ party Mermin-Ardehali-Belinskii-Klyshko (MABK) and Werner-Wolf-Weinfurter-\.Zukowski-Brukner (WWW\.ZB) family of linear Bell inequalities. Moreover, our methodology enables the recovery of self-testing statements for  Bell functions which are not only the mean value of a linear operator. Specifically, to showcase the versatility of our methodology, we obtain self-testing statements for the maximal violation of $N$ party Uffink's quadratic Bell inequalities, which form tight witnesses of genuine multipartite non-locality \cite{2002}. 

The paper is organized as follows. In section \textbf{II} we present the requisite preliminaries and specify the families of multipartite Bell inequalities we consider in this article. Section \textbf{III}, we develop the mathematical preliminaries of the self-testing scheme. In particular, we show that the observables of each party in two setting binary outcome multipartite Bell scenarios can always be simultaneously represented as anti-diagonal matrices. In Section \textbf{IV}, we utilize this anti-diagonal matrix representation to obtain self-testing statements for the maximal violations of $N$ party MABK and tripartite WWW\.ZB family of linear (on correlators) Bell inequalities. While the former serves to certify our scheme, the latter demonstrates the spectrum of situations one can encounter. Finally, in Section \textbf{IV} we obtain self-testing statements for the maximal violation of $N$ party Uffinik's quadratic Bell inequalities. In section \textbf{V}, we conclude by providing a brief summary of our work and specifying potential applications.

\section{Two-setting $N$-party correlation Bell inequalities}
\label{INEQS}
This section presents the requisite preliminaries and specifies the families of Bell inequalities considered in this article. Specifically, here we consider multipartite Bell scenarios entailing $N$ spatially separated (hence non-signaling) parties. We restrict ourselves to Bell scenarios where each party $j\in\{1,\ldots, N\}$ has two binary outcome observables $A^{(j)}, A'^{(j)}$ for simplicity and brevity \footnote{However, our technique can easily be applied to proving more complex scenarios.}. In contrast to the well-studied bipartite Bell scenarios, multipartite scenarios are substantially richer in complexity. While the notion of multipartite locality is an obvious extension of bipartite locality, multipartite behaviors can be non-local in many distinct ways. Apart from this, the most significant impediment in obtaining self-testing statements for multipartite Bell scenarios is that they do not admit simplifying characterizations such as Schmidt decomposition, unlike the bipartite Bell scenarios.

Typically, Bell inequalities comprise of a Bell expression and a corresponding local causal bound. The violation of Bell inequalities witnesses the non-locality of the underlying behaviors. 
 In the rest of this section, we introduce the families of multipartite Bell inequalities, for which we demonstrate self-testing statements using our proof technique. 
\subsection{ Linear inequalities}
The most frequently used Bell-inequalities comprise of Bell expressions which are linear functionals of observed probabilities. Our self-testing argument is immediately applicable for any Bell inequality whose Bell expression is the mean of any linear combination of $N$ party operators of the form 
${\bigotimes_{i=1}^N O^{(i)}}$, where $O^{(i)}\in\{A^{(i)},A'^{(i)}\}$. Among these linear Bell inequalities, we consider the \emph{Werner-Wolf-Weinfurter-\.Zukowski-Brukner} (WWW\.ZB) families of correlation Bell inequalities for which there exists a well-defined systematic characterization \cite{1099690194}. The Bell operator for WWW\.ZB family of Bell inequalities has the following general form,
\begin{eqnarray} \label{WWZB}
\mathcal{W}_{N}=\frac{1}{2^N}\sum_{s_1,...,s_N=\pm 1}S(s_1,...,s_N)\bigotimes_{j=1}^N(A^{(j)}+s_jA'^{(j)}),
\end{eqnarray}
where $S(s_1,...s_N)=\pm 1$. Every tight (WWW\.ZB) inequality is given by $\langle \mathcal{W}_{N}\rangle \leq_{\mathcal{L}}1$. Where $\mathcal{L}$ stands for the local hidden variable polytope. 

Next, we introduce a sub-family of  (WWW\.ZB) Bell inequalities the \emph{Mermin-Ardehali-Belinskii-Klyshko} (MABK) family of inequalities, featuring one inequality for any number $N$ of parities \cite{Belinski__1993, PhysRevA.46.5375, PhysRevLett.65.1838}. Moreover, in the following subsection, we introduce non-linear (quadratic) Bell inequalities composed of these MABK inequalities. The $N$-party MABK operators can be obtained recursively as, 
\begin{eqnarray} \nonumber\label{MABK}
& \mathcal{M}_N &=\frac{1}{2}\Big(\mathcal{M}_{N-1}\otimes(A^{(j)}+A'^{(j)}) \\
& & \hspace{21pt} + \mathcal{M'}_{N-1}\otimes(A^{(j)}-A'^{(j)})\Big)\nonumber \\ 
& & = \frac{1}{2}\Big((\frac{1-i}{{2}})^{N-1}\bigotimes_{j=1}^N(A^{(j)}+iA'^{(j)})\nonumber\\
 & &\hspace{21pt} + (\frac{1+i}{{2}})^{N-1}\bigotimes_{j=1}^N(A^{(j)}-iA'^{(j)})\Big),
\end{eqnarray}
where $\mathcal{M}'_{N-1}$ is the same expression as $\mathcal{M}_{N-1}$ but with all
$A^{(j)}$ and $A'^{(j)}$ interchanged.
 The corresponding Bell inequalities are of the form, $\langle \mathcal{M}_{N}\rangle \leq_{\mathcal{L}}1\leq_{\mathcal{Q}_{N-1}}2^{\frac{N-2}{2}} \leq_{\mathcal{Q}}2^{\frac{N-1}{2}},$ where $\mathcal{L},\mathcal{Q}_{N-1},\mathcal{Q}$ stand for the convex polytope of local hidden variable correlations, the convex set of biseparable quantum correlations, and the convex set of quantum behaviors, respectively.
 The maximal quantum value can only be attained when the local observables anti-commutes and the parties share the maximally entangled $N$-partite GHZ state.

Finally, we present yet another relevant sub-family of the (WWW\.ZB) inequalities, referred to as \emph{Svetlichny} inequalities \cite{2002}, which were explicitly conceived to witness genuine $N$-partite non-locality. Moreover, the quadratic Bell inequalities featured in the following subsection can also be composed of $N$-partite Svetlichny-Bell inequalities. The  Svetlichny operator can be formed of MABK operators \eqref{MABK} in the following way,

\begin{equation} \label{Svet}
    \mathcal{S}^{\pm}_{N} = \begin{cases} \displaystyle
    2^{k-1}((-1)^\frac{k(k\pm1)}{2}\mathcal{M}_N^{\pm}, \\ (\text{for } N=2k); \\
    2^{k\pm1}((-1)^\frac{k(k\pm1)}{2}\mathcal{M}_N\mp (-1)^\frac{k(k\mp1)}{2}\mathcal{M'}_N)), \\ (\text{for } N=2k+1). 
    \end{cases}
\end{equation}
The notation $ \mathcal{M}^{+}_N$ is equivalent to $ \mathcal{M}_N$  and $ \mathcal{M}^{-}_N$ is equivalent to $ \mathcal{M'}_N$. The corresponding Svetlichny inequalities are of the form, $
   \langle\mathcal{S}^{\pm} _{N}\rangle \leq_{\mathcal{Q}_{N-1}}2^{N-1}\leq_{\mathcal{Q}}2^{N-\frac{1}{2}}. 
$
\begin{figure}
    \centering
    \includegraphics[width=\linewidth]{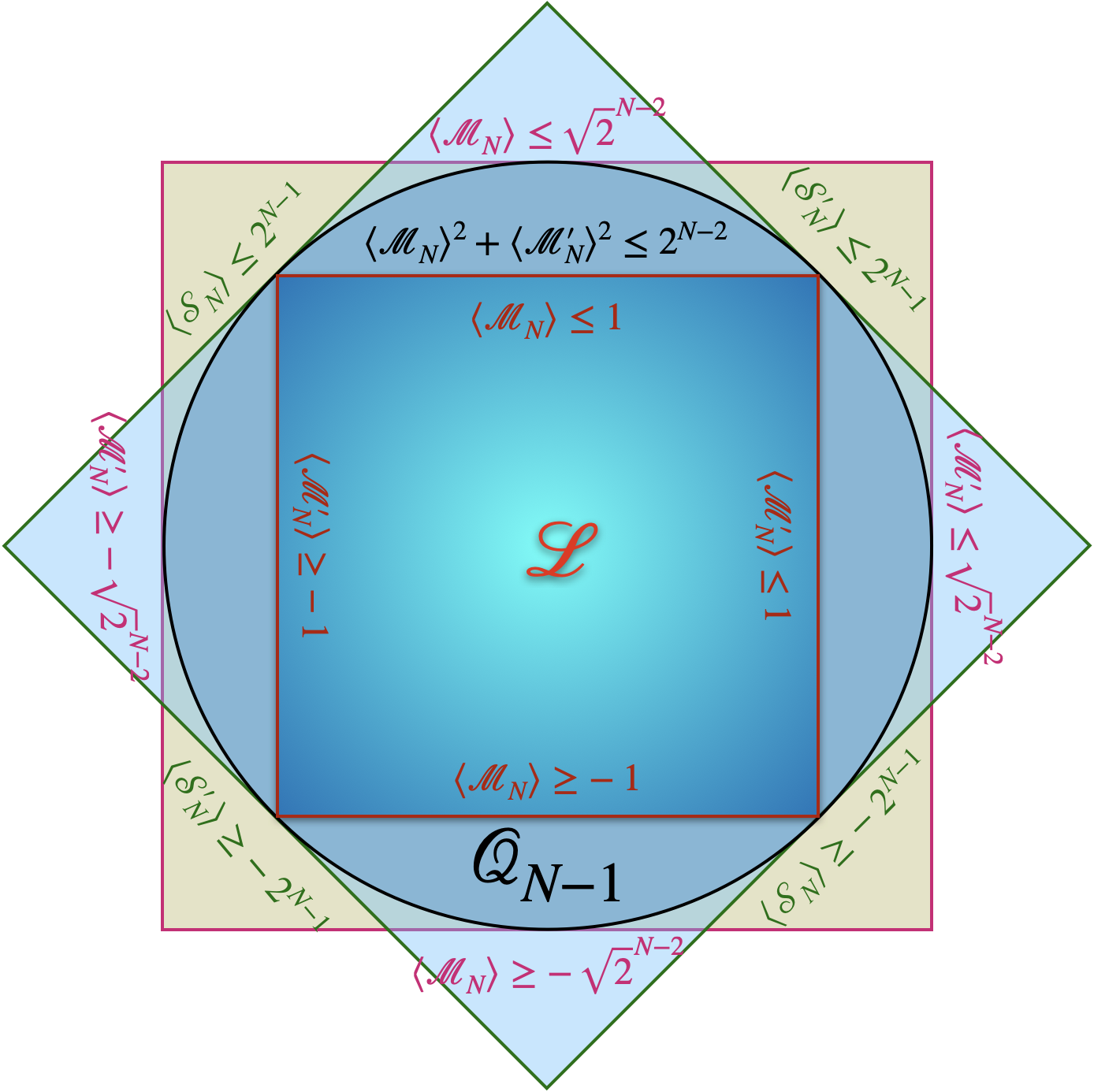}
    \caption{This graphic is a schematic representation of the correlations in multipartite (involving arbitrary number $N$ of spatially separated parties) Bell scenarios. Just like the bipartite Bell scenarios, the correlations which admit local hidden variable explanations form a convex polytope $\mathcal{L}$ (shiny blue small horizontal square), whose facets are the $N$-party MABK inequalities \eqref{MABK} (red edges). However, the convex set of biseparable quantum correlations $\mathcal{Q}_{N-1}$ (sky blue disk) does not form a polytope. Consequently, the linear inequalities such as the MABK \eqref{MABK} (pink edges of the large horizontal square) and Svetlichny inequalities \eqref{Svet} (green edges of the tilted square) do not form tight witness of genuine multipartite quantum nonlocality. As the boundary of biseparable quantum correlations (black circle) is non-linear, the Uffink's quadratic inequalities \eqref{Uff} form tighter witnesses of genuine multipartite quantum nonlocality.    
    \label{manipuraka}}
\end{figure}
 
\subsection{Quadratic inequalities}
In bipartite Bell scenarios, without loss of generality, it is enough to consider linear Bell inequalities to witness non-locality, as the set of behaviors which admit a local casual explanation form a convex polytope. In contrast to the bipartite case, non-linear Bell inequalities form tighter witnesses of genuine multipartite non-locality. In multipartite Bell scenarios, the convex set of bi-separable quantum behaviors does not form a polytope. 

In particular, the Uffink's quadratic Bell inequalities form stronger witnesses of genuine multipartite non-locality than the linear inequalities. There are two distinct families of $N$ party quadratic Bell inequalities formed of the MABK \eqref{MABK} and Svetlichny \eqref{Svet}, families of linear inequalities, $\mathcal{U}^{\mathcal{M}}_N$ and  $\mathcal{U}^{\mathcal{S}}_N$,   respectively, which have the form
\begin{equation}\label{Uff}
\mathcal{U}^{\mathcal{M}}_N=\mean{\mathcal{M}_{N}}^2+\mean{\mathcal{M}'_{N}}^2\leq_{\mathcal{Q}_{N-1}}2^{N-2}\leq_{\mathcal{Q}}2^{N-1} ,
\end{equation}
\begin{equation}\label{UffS}
\mathcal{U}^{\mathcal{S}}_N=\mean{\mathcal{S}_{N}}^2+\mean{\mathcal{S}'_{N}}^2\leq_{\mathcal{Q}_{N-1}}2^{2N-2}\leq_{\mathcal{Q}} 2^{2N-1},
\end{equation}

\section{Characterizing local observables}
\label{maths}
In this section, we obtain a characterization for the binary outcome local observables $A=A^{(j)}$ and $A'=A'^{(j)}$ on an arbitrary Hilbert space $\mathcal{H}=\mathcal{H}^j$ for any given party $j\in\{1,\ldots,N\}$. In the light of Naimark's dilation theorem, without loss of generality we can take the observables $A$ and $A'$ to be projective, i.e., $A^2=A'^2=\mathbb{I}$. As for Bell operators composed of projective observables the optimal state is simply the eigenvector corresponding to the maximum eigenvalue, the shared state can always be taken to be pure. First, via the following lemma we demonstrate that second observable $A'$ can be split into two observables, one which commutes with the first observable $A$ and one which anti-commutes with $A$,
\begin{lemma}
\label{BLOCK2}
Given any two binary outcome projective observables $A$ and $A'$, $A'$ can be decomposed as the sum of two observables $-\mathds{1}\leq A'_{-} \leq \mathds{1}$ and $-\mathds{1}\leq A'_{+} \leq \mathds{1}$, such that   $\comm{A}{A'_{+}}=0$, $\acomm{A }{A'_{-}}=0$, $\acomm{A'_{+}}{A'_{-}}=0$ and $(A'_{+})^2+(A'_{-})^2=\mathds{1}$.
\end{lemma}
\begin{proof}
Without loss of generality, the binary outcome projective observable $A$ can be represented as a diagonal matrix with positive and negative eigenvalues grouped together,
\begin{equation}
\label{UNITS}
A=\left(\begin{array}{cc}\openone_{m} & 0 \\ 0 &-\openone_{n}\end{array}\right),
\end{equation}
where $\mathds{1}_m$ is the $m\times m$ identity operator. With respect to $A$, the binary outcome projective observable $A'$ has the following generic matrix representation,
\begin{equation}
\label{UNITS2}
    A'=\left(\begin{array}{cc}D_1 & D_2\\ 
    D_2^\dagger & D_3\end{array}\right).
\end{equation}
such that,
\begin{align}
\label{UNITS3}
A'_{-} & =\left(\begin{array}{cc}
            0 & D_2\\
            D_2^\dagger & 0 
          \end{array}\right)\,,\\
\label{UNITS4}
A'_{+}&=\left(\begin{array}{cc}
        D_1 & 0 \\
        0 & D_3
        \end{array}\right)\,.
\end{align}
Clearly, $\comm{A}{A'_{+}}=0$, and $\acomm{A }{A'_{-}}=0$.
As $A'$ is projective we have,
\begin{align}
    (A')^2 = &(A'_{+})^2+(A'_{-})^2 + \acomm{A'_{+} }{A'_{-}} \\
            =&\left(\begin{array}{cc}
                D_1^2+D_2D_2^\dagger & D_1D_2+D_2D_3\\
                D_2^\dagger D_1+D_3D_2^\dagger & D_2^\dagger D_2+D_3^2
            \end{array}\right)\nonumber\\
    =&\,\mathds{1}\,,
\end{align}
which requires the off-diagonal blocks to be zero,
\begin{eqnarray}
    &\{A'_{+},A'_{-}\}&=\left(\begin{array}{cc}
        0 & D_1D_2+D_2D_3\\
        D_2^\dagger D_1+D_3D_2^\dagger & 0 
    \end{array}\right)\nonumber\\
    & &=0,
\end{eqnarray}
and leaves $(A'_{+})^2+(A'_{-})^2=\mathds{1}$, 
which completes the proof.
\end{proof}
Using the above Lemma, the relation between the eigenvalues of $A'$, $A'_{-}$, and $A'_{+}$ can be ascertained to be $\lambda_{A'}^i=\pm\sqrt{(\lambda^i_{A'_{-}})^2 + (\lambda^i_{A'_{+}})^2}=\pm 1$, where $\lambda^i_{\cdot}$ denotes the $i^\text{th}$ eigenvalue \cite{Ponomarenko2020EigenvaluesOT}. So that the eigenvalues can be without loss of generality taken as $\pm \sin{\theta_i}$ and $\pm \cos{\theta_i}$ for $A'_{-}$ and $A'_{+}$ respectively. 

Next, we show that the spectra of any two anti-commuting observables must be symmetric, and each observable maps the positive eigenspaces to the negative eigenspaces of the other.

\begin{lemma}
\label{spectrum}
Given any two anti-commuting observables $A$ and $A'_{-}$, their spectra are symmetric (i.e., if $\lambda$ is an eigenvalue of one of the operators, then so is $-\lambda$), moreover, $A'_{-}$ ($A$) is a linear map between positive and negative eigenspaces of $A$ ($A'_{-}$) or nullifies its eigenvectors.
\end{lemma}
\begin{proof}Let $\ket{\psi}$ be an eigenvector of $A$ corresponding to the eigenvalue $\lambda$ such that 
$A\ket{\psi}=\lambda\ket{\psi}$, then
\begin{equation}
A(A'_{-}\ket{\psi})=-A'_{-}A\ket{\psi}=-\lambda A'_{-}\ket{\psi}.
\end{equation}
Thus $A'_{-}\ket{\psi}$ is either a null vector or a new eigenvector of $A$ with the sign of the eigenvalue flipped. Similarly, one can show the same for eigenvalues of $A'_{-}$.
\end{proof}
The above lemma further yields the following observations,
\begin{corollary}
Given any two anti-commuting observables $A$ and $A'_{-}$, such that $A^2=\mathds{1}$, then $A$  provides a bijective mapping between the eigenspaces $E^{A'_{-}}_{\lambda}$ and $E^{A'_{-}}_{-\lambda}$ corresponding to the eigenvalues $ \lambda^{A'_{-}}$ and $- \lambda^{A'_{-}}$, respectively. Moreover, the direct sum of the pair of eigenspaces $E^{A'_{-}}_{\pm\lambda_i}=E^{A'_{-}}_{\lambda_i}\oplus E^{A'_{-}}_{-\lambda_i}$ is even dimensional, and so is the subspace $\bigoplus_{i}E^{A'_{-}}_{\pm\lambda_i}$.
\end{corollary}

\begin{proof}
If $\lambda^{A'_{-}}$ is an eigenvalue of $A'_{-}$, with eigenvector $\ket\psi$, then,
\begin{equation}
    AA'_{-}\ket{\psi} = -A'_{-}A\ket\psi = \lambda^{A'_{-}}A\ket\psi,
\end{equation}
and applying $A$ to the eigenvalue equation again,
\begin{equation}
    -AA'_{-}A\ket{\psi} = \lambda^{A'_{-}}A^2\ket\psi = \lambda^{A'_{-}}\ket\psi,
\end{equation}
we recover the same eigenvector again. In other words, $A$ maps the subspace $E^{A'_{-}}_{\lambda_i}\oplus E^{A'_{-}}_{-\lambda_i}$ to itself. This pairing of eigenvectors leads us to the conclusion that this subspace is even dimensional.
The diagonalizability of the operator, $A'_{-}$, implies that the effective Hilbert space can be decomposed into the direct sum of its eigenspaces, 
\begin{equation}
\label{eq:direct-sum}
    \mathcal{H} = E^{A'_{-}}_{\lambda_1}\oplus E^{A'_{-}}_{-\lambda_1}\cdots\oplus E^{A'_{-}}_{\lambda_r}\oplus E^{A'_{-}}_{-\lambda_r}\oplus\ker(A'_{-}).
\end{equation}
It also follows from the anti-commutation relation that $A A'_{-} A^{-1}=-A'_{-}$, and thus
\begin{equation}
    \det(A'_{-})=(-1)^d\det(A'_{-}),
\end{equation}
where $d=n+m$ is the dimension of the Hilbert space. 
This is only possible if $\det(A'_{-})=0$ or $A'_{-}$ is even dimensional.
In the case when $\det(A'_{-})=0$, the null subspace of $A'_{-}$, $\ker(A'_{-})$, can be ignored as it does not  contribute to the Bell inequality violation. This implies that the effective subspace of such operators can be truncated to   an even-dimensional subspace. Finally, as $A'=A'_{+}+A'_{-}$, the effective subspaces of $A'_{+}$ and $A'$ are also be truncated to the same even-dimensional subspace as $A'_{-}$.
\end{proof}
Moreover, the lemma \ref{spectrum} yields the following succinct parameterization of the two mutually anti-commuting components of any projective observable,
\begin{corollary}
Given two even-dimensional anti-commuting operators $A_{+}$ and $A_{-}$ such that $(A'_{+})^2+(A'_{-})^2=\mathds{1}$, they can be written in the form,
\begin{align}
A'_{+|\theta_i}=\sin{\theta_i}B_{+|\theta_i} & & A'_{-|\theta_i}=\cos{\theta_i}B_{-|\theta_i}
\end{align}
 when restricted to the subspace $E^{A'_{-}}_{+\theta_i}\oplus E^{A'_{-}}_{-\theta_i}$ with corresponding eigenvalues $\pm\sin(\theta_i)$ for $A_{+}$ and $\pm\cos(\theta_i)$ for $A_{-}$, such that the operators $B_{\pm|\theta_i}$ are traceless and projective.
\end{corollary}
\begin{proof}
We have already shown the dimension of the combined eigenspaces $E^{A'_{-}}_{\lambda_i}\oplus E^{A'_{-}}_{-\lambda_i}$ is even. 
Then the two anti-commuting operators $A'_{+}$ and $A'_{-}$ are restricted to this eigenspace, denoted by $A'_{+\vert\theta_i}$ and $A'_{-\vert\theta_i}$, will only have eigenvalues $\pm\sin{\theta_i}$ and $\pm\cos{\theta_i}$ respectively, while still satisfying $(A'_{+|\theta_i})^2+(A'_{-|\theta_i})^2=\mathds{1}$.
It is then possible to write the same relation in terms of scaled operators,
as 
\begin{equation}
    \cos^2{\theta_i} B_{+|\theta_i}^2 + \sin^2{\theta_i} B_{-|\theta_i}^2 = \mathds{1}.
\end{equation}
As a result the operators $B_{\pm|\theta_i}$ have eigenvalues $\pm1$ that occur in pairs. Therefore, $\Tr(B_{\pm|\theta_i})=0$ and $B_{\pm|\theta_i}^2=\mathds{1}$.
\end{proof}
Using the above results, in such an even-dimensional subspace the following lemma holds.
\begin{lemma}
\label{ANTICOM}
Given any two traceless and projective anti-commuting observables $B_+$ and $B_{-}$, the observable  $\cos{\alpha}B_+ +\sin\alpha B_{-}$ is also traceless and projective.

\end{lemma}
\begin{proof}
Expanding the square,
\begin{align}
(\cos\alpha & B_+ + \sin\alpha B_{-} )^2\nonumber\\
=&\cos^2\alpha B_{+}^2 + \sin^2\alpha B_{-}^2 + \cos\alpha\sin\alpha\acomm{B_{+}}{B_{-}} \nonumber\\
=&\mathds{1}.
\end{align}
And the trace is simply,
\begin{equation} 
\cos\alpha\Tr(B_{+})+\sin\alpha\Tr(B_{-})=0.
\end{equation}
\end{proof}

The above lemma shows that whenever any inequality is maximally violated the effective local dimension of any subsystem must be even dimensional. Now, we present the main ingredient of our self-testing proof technique, i.e., anti-diagonal representation for local observables of each party up to local isometries. 

\begin{theorem}
\label{3TM}
Given any three binary outcome traceless and projective observables  $A$, $B_{+}$, and $B_{-}$, such that $\comm{A}{B_{+}}=0$, $\acomm{A }{B_{-}}=0$, and $\acomm{B_{+}}{B_{-}}=0$, then these operators have a simultaneous anti-diagonal matrix representation.  
\end{theorem}

\begin{proof}
As $\comm{A}{B_{+}}=0$ we can take the dimension of the subspace for which the eigenvalues of $A$ and $B_{+}$ are equal to be $2d_1$, and $2d_2$ for the subspace where the eigenvalues differ. Consequently, without loss of generality, by Lemma \ref{spectrum} the operators $A$, $B_{+}$, and $B_{-}$ have the following matrix representations,
\begin{eqnarray}
    A=&&\left(\begin{array}{cccc}
            \mathds{1}_{d_1} & \cdot & \cdot & \cdot\\
            \cdot & \mathds{1}_{d_2} & \cdot & \cdot\\
            \cdot & \cdot & -\mathds{1}_{d_2} & \cdot \\
            \cdot & \cdot & \cdot & -\mathds{1}_{d_1}
    \end{array}\right),\nonumber\\
    B_{+}=&&\left(\begin{array}{cccc}
            \mathds{1}_{d_1} & \cdot & \cdot & \cdot \\
            \cdot & -\mathds{1}_{d_2} & \cdot & \cdot \\
            \cdot & \cdot & \mathds{1}_{d_2}& \cdot \\
            \cdot & \cdot & \cdot & -\mathds{1}_{d_1}
    \end{array}\right),\nonumber\\
    B_{-}=&&\left(\begin{array}{cccc}
            \cdot & \cdot & \cdot & U_1\\
            \cdot & \cdot & U_2& \cdot \\
            \cdot & U_2^\dagger & \cdot & \cdot \\
            U_1\dagger & \cdot & \cdot & \cdot 
    \end{array}\right).
\end{eqnarray}
Since $B_{-}$ is projective, $U_1$ and $U_2$ must be unitary. Thus, without altering $A$ or $B_{+}$ we can take four unitaries $V_1$, $V_2$, $V_3$, $V_4$, each acting on a different block, such that $V_1U_1V_4^\dagger=J_{d_1}$ and $V_2U_2V_3^\dagger=J_{d_2}$, where $J_{d}$ is the row-reversed $d\times d$ identity matrix. 

Consequently, we can now restrict ourselves to considering any one of the $d_1+d_2$ two-dimensional subspaces, on which $A$ and $B_{+}$ are represented by $\pm\sigma_z$, while $B_{-}$ is projected onto $\sigma_x$.
Finally, in each of these subspaces applying a rotation by $\frac{2\pi}{3}$ with respect to axis $(1,1,1)$ yields, 
\begin{equation}
U=\frac{1}{2}\left(\begin{array}{cc}-1+\iota&1+\iota\\-1+\iota&-1-\iota\end{array}\right),    
\end{equation}
transforming $\sigma_z\rightarrow\sigma_x\rightarrow\sigma_y$, and bringing all three operators to strictly antidiagonal form. 
\end{proof}
Summarizing, the theorem \ref{3TM} along with lemmas \ref{spectrum} and \ref{ANTICOM} allow us to take the first observable of each party $A^{(j)}$ to be equivalent to $\sigma_x$ on all relevant two dimensional subspaces, while the second operator $A'^{(j)}$ can be taken to be $A'^{(j)}=\cos{\theta_{j}}\sigma_x+\sin{\theta_{j}}\sigma_y$.
We would like to stress the key significance of this theorem. As the Bell operators considered in this article are tensor products of local observables for each party, and via theorem \ref{3TM} each local observable can have anti-diagonal matrix representation, the Bell operator itself has anti-diagonal matrix representation. Consequently, only the states with maximal (in modulo) antidiagonal elements, or their mixtures, can provide the maximal contribution to the Bell violation. In particular, this implies the maximum violation two setting $N$-partite correlation Bell inequalities necessitates $N$-qubit GHZ states, or their mixtures, as we demonstrate in the following section.  
\section{Self-testing statements for multipartite inequalities }
\label{ST3W}
In this section, we use the tools developed in the previous section to obtain self-testing statements for linear MABK family \eqref{MABK} of $N$ party inequalities, and sketch the proofs for the self-testing statements for all distinct equivalence classes of tripartite WWW\.ZB facet inequalities \eqref{WWZB}. Finally, we obtain self-testing statements for Uffnick's family \eqref{Uff} of non-linear $N$ party inequalities. 
\subsection{$N$ party MABK inequalities}
\begin{theorem}\label{MABKTheo}
In order to achieve maximal quantum violation of a $N$-party MABK inequality, $\langle \mathcal{M}_{N}\rangle=2^{\frac{N-1}{2}}$, the parties must share a $N$ qubit GHZ state $\ket{GHZ_N}=\frac{1}{\sqrt{2}}(\ket{0}^{\otimes N}+ e^{\iota \phi_{ N}}\ket{1}^{\otimes N})$ and perform maximally anti-commuting projective measurements $A^{(j)}=\sigma_x$ and $A'^{(j)}=\sigma_y$ (upto local auxiliary systems and local isometries). 
\end{theorem}
\begin{proof}
From lemma \ref{spectrum} and theorem \ref{3TM}, without loss of generality, the local observables of any party $j$ can be taken to be, 
\begin{eqnarray} \nonumber
    & A^{(j)}&=\sigma_x, \\
    & A'^{(j)}&=\cos{\theta_{j}}\sigma_x+\sin{\theta_{j}}\sigma_y,
\end{eqnarray}
acting on the effective two dimensional subspace. This parametrization implies that the MABK operator \eqref{MABK} $\mathcal{M}_{N}$ has the following anti-diagonal matrix representation,
\begin{eqnarray} \label{MABKMatrix}
 \mathcal{M}_{N}  =\text{adiag} \begin{pmatrix}
\frac{1}{2}\Big((\frac{1-\iota}{2})^{N-1}\Pi_{j=1}^{N}(1+\iota e^{-\iota\theta_j})\\
\hspace{21pt}+(\frac{1+\iota}{2})^{N-1}\Pi_{j=1}^{N}(1-\iota e^{-\iota \theta_j})\Big)\\
\vdots\\
\frac{1}{2}\Big((\frac{1-\iota}{2})^{N-1}\Pi_{j=1}^{N}(1+\iota e^{\iota\theta_j})\\
\hspace{21pt}+(\frac{1+\iota}{2})^{N-1}\Pi_{j=1}^{N}(1-\iota e^{\iota \theta_j})\Big)
\end{pmatrix},
\end{eqnarray}
where $\text{adiag}$ represents a matrix with non-zero values only on the anti-diagonal. It is easy to see for any combination, $\forall j\in\{1,\ldots,N\}:\theta_j=\pm\frac{\pi}{2}$, one of the anti-diagonal element attains the maximum absolute value of $2^{\frac{N-1}{2}}$ while the others vanish.  

Let us, for simplicity, consider the subspace spanned by $\{\ket{0}^{\otimes N}, \ket{1}^{\otimes N}\}$. On the subspace we consider the state $\ket{\psi_N}=\alpha\ket{0}^{\otimes N}+ \beta\ket{1}^{\otimes N}$ which effectively yields a weighted sum of the top and bottom anti-diagonal elements of the matrix \eqref{MABKMatrix}. Consequently, the expectation value of the Hermitian operator $\mathcal{M}_N$ for the state $\ket{\psi_N}$ has the expression,
\begin{eqnarray} \nonumber
&\bra{\psi_N}\mathcal{M}_{N}\ket{\psi_N}
= &2\Re\{\overline{\alpha}\frac{1}{2}\Big((\frac{1-\iota}{2})^{N-1}\Pi_{j=1}^{N}(1+\iota e^{-\iota\theta_j}) \\
& & +(\frac{1+\iota}{2})^{N-1}\Pi_{j=1}^{N}(1-\iota e^{-\iota \theta_j})\Big)\beta\},
\end{eqnarray}
Using the fact that $\forall \alpha,\beta \in \mathbb{C}: \ \Re{\alpha \beta}\leq |\alpha \beta|= |\alpha||\beta|$ we bound $\bra{\psi_N}\mathcal{M}_{N}\ket{\psi_N}$ from above in the following way,
\begin{eqnarray} \nonumber
&\bra{\psi_N}\mathcal{M}_{N}\ket{\psi_N} \leq & 2|\overline{\alpha}||\beta||\frac{1}{2}\Big((\frac{1-\iota}{2})^{N-1}\Pi_{j=1}^{N}(1+\iota e^{-\iota\theta_j}) \nonumber \\
& & +(\frac{1+\iota}{2})^{N-1}\Pi_{j=1}^{N}(1-\iota e^{-\iota \theta_j})\Big)|. 
\end{eqnarray}
As $|\alpha|^2+|\beta|^2=1$, the maximum value of the above expression can only be attained for $|\alpha|=|\beta|=\frac{1}{\sqrt{2}}$, which picks out $\ket{GHZ_N}$ as the shared state. Consequently, we retrieve the following upper bound, 
\begin{eqnarray}
&\bra{\psi_N}\mathcal{M}_{N}\ket{\psi_N}\leq & |\frac{1}{2}\Big((\frac{1-\iota}{2})^{N-1}\Pi_{j=1}^{N}(1+\iota e^{-\iota\theta_j})\nonumber \\
 & & +(\frac{1+\iota}{2})^{N-1}\Pi_{j=1}^{N}(1-\iota e^{-\iota \theta_j})\Big)|.
\end{eqnarray}
Now, as $|(\frac{1+\iota}{2})^{N-1}|=|(\frac{1-\iota}{2})^{N-1}|=\frac{1}{2^{\frac{N-1}{2}}}$ and $(1+\iota e^{-\iota \theta_j})=2\cos{(\frac{\pi}{4}+\frac{\theta_j}{2})}$ and $(1-\iota e^{-\iota \theta_j})=2\sin{(\frac{\pi}{4}+\frac{\theta_j}{2})}$, we can further upper bound the above expression as,
\begin{eqnarray}
&\bra{\psi_N}\mathcal{M}_{N}\ket{\psi_N} \leq & 2^{\frac{N-1}{2}}\Big(\Pi_{j=1}^{N}|\cos{\frac{\pi}{4}+\frac{\theta_j}{2}}|\nonumber \\
& &\hspace{25pt} + \Pi_{j=1}^{N}|\sin{\frac{\pi}{4}+\frac{\theta_j}{2}}|\Big) \nonumber.
\end{eqnarray}
Dropping positive terms corresponding to any $N-2$ parties we retrieve the simplified upper bound,
\begin{eqnarray}
&\bra{\psi_N}\mathcal{M}_{N}\ket{\psi_N} & \leq  \nonumber {2^{\frac{N-1}{2}}}\Big(\left|\sin\left(\frac{2\theta_{i}+\pi}{4}\right)\sin\left(\frac{2\theta_{j}+\pi}{4}\right)\right| \\ 
&   & \hspace{21pt} + \left|\cos\left(\frac{2\theta_{i}+\pi}{4}\right)\cos\left(\frac{2\theta_{j}+\pi}{4}\right)\right|\Big)\nonumber \\ 
&  & =  2^{\frac{N-1}{2}} \max\Big\{\Big|\cos(\frac{\theta_{i}+\theta_{j}+\pi}{2})\Big|,\nonumber\\
& &  \hspace{27pt} \Big|\cos(\frac{\theta_{i}-\theta_{j}}{2}n)\Big|\Big\} \nonumber \\
& & \leq  2^{\frac{N-1}{2}},
\end{eqnarray}
As the choice of $i,j \in \{1,\ldots,N\}$ is completely arbitrary, 
the inequality can only be saturated when $\forall j\in\{1,\ldots,N\}:\theta_j=\pm\frac{\pi}{2}$, i.e., for each party the local observables maximally anti-commute and the state that maximally violate the quntum bound is $\ket{GHZ_N}$ up to auxiliarly degrees of freedom on which the measurements act trivially, and local basis transformations. This implies that the actual shared state could be of the form $\ket{GHZ_N}\bigotimes\ket{\Psi}$ where the arbitrary state $\ket{\Psi}$ on auxiliary degrees of freedom does not contribute to the operational Bell violation, and thus is referred to as the junk state.

\end{proof}

The proof technique readily applies to the Svetlichny family of $N$-party inequalities as they are composed of the $N$-party MABK inequalities. Furthermore, the proof technique enables self-testing of a much broader class of $N$-party WWW\.ZB inequalities. To demonstrate this, we sketch the proofs for self-testing of tripartite  WWW\.ZB inequalities \eqref{WWZB}. 

\subsection{Tripartite  WWW\.ZB inequalities }\label{WWWZB1}
For the tripartite scenarios where each party has two binary outcomes observables the local real polytope is characterized by $256$ facet correlation inequalities which can be grouped into four equivalent classes of non-trivial inequalities, 

The first class is composed of correlation inequalities equivalent (up-to relabeling) to the \textit{Mermin's inequality} $\mean{\mathcal{M}_3}\leq1$,
\begin{eqnarray}
& \frac{1}{2}\Big(&\mean{(A^{(1)}A^{(2)}A'^{(3)}}+\mean{A^{(1)}A'{(2)}A^{(3)}}\nonumber \\
& &+\mean{A'^{(1)}A^{(2)}A^{(3)}}-\mean{A'^{(1)}A'^{(2)}A'^{(3)})}\Big)\leq 1. \nonumber
\end{eqnarray}
For these inequalities the proof of Theorem \ref{MABKTheo} directly applies. As a consequence, we retrieve the tripartite GHZ state $\ket{GHZ_3}$ as well as maximally anti-commuting local observables as necessary ingredients for the maximal quantum violation 2.   

The second equivalence class of tripartite inequalities is that of  \textit{unbalanced inequalities}, which have the following generic form.
 \begin{eqnarray}
 & \Big(&3\mean{A^{(1)}A^{(2)}A^{(3)}}+\mean{A^{(1)}A^{(2)}A'^{(3)}}\nonumber\\
 & &+\mean{A^{(1)}A'^{(2)}A^{(3)}}+\mean{A'^{(1)}A^{(2)}A^{(3)}}\nonumber\\
 & & -\mean{A^{(1)}A'^{(2)}A'^{(3)}}-\mean{A'^{(1)}A^{(2)}A'^{(3)}}\nonumber\\
 & &-\mean{A'^{(1)}A'^{(2)}A^{(3)}}+\mean{A'^{(1)}A'^{(2)}A'^{(3)}}\Big)\leq 4.\nonumber
 \end{eqnarray}
 The corresponding operator has a anti-diagonal matrix representation with values of the form,
 \begin{eqnarray}
 4+(-1+e^{\iota k_1\theta_{1}})(-1+e^{\iota k_2\theta_{2}})(-1+e^{\iota k_3\theta_{3}}),
 \end{eqnarray}
Where $(k_1,k_2,k_3\in\{-1,1\})$. Unlike the MABK class of inequalties the absolute values of these anti-diagonal terms are maximised when $\cos\theta_{j}=-\frac{1}{3}$. Thus, yet again, the tripartite GHZ state $\ket{GHZ_3}$ is distinguished, but non maximally anti-commuting operators are required for maximal quantum violation of these inequalities is $\frac{20}{3}$. 
 
 Next, we have the equivalence class of \textit{extended CHSH inequalities} of the form, 
\begin{eqnarray}
&\frac{1}{2}\Big(&\mean{A^{(1)}A^{(2)}A^{(3)}}+\mean{A'^{(1)}A^{(2)}A^{(3)}}\nonumber\\
& &+\mean{A^{(1)}A'^{(2)}A'^{(3)}}-\mean{A'^{(1)}A'^{(2)}A'^{(3)}}\Big)\leq 1\nonumber.
\end{eqnarray}
For these inequalities, the operator in the anti-diagonal matrix representation has elements of the form,  
\begin{eqnarray}
\frac{1}{2}(1+e^{\iota k_1\theta_{1}}+e^{\iota (k_2\theta_{2}+k_3\theta_{3})}-e^{\iota(k_1\theta_{1}+k_2\theta_{2}+k_3\theta_{3}))}.
\end{eqnarray}
Consequently,the corresponding absolute values are $2\sqrt{1\pm\sin\theta_{1}\sin(k_2\theta_{2}\pm k_3\theta_{3})}$. Clearly, for the maximal quantum violation the operators $A^{(1)}$ and $A'^{(1)}$ must maximally anti-commute, i.e., $\theta_{1}=\pm\frac{\pi}{2}$. However, the maximal quantum violation only requires the sum  $k_2\theta_{2}\pm k_3\theta_{3}=\pm\frac{\pi}{2}$, i.e., the optimal $A'^{(2)}$ is defined only in reference to $A'^{(3)}$. Clearly, for the maximal quantum violation $2\sqrt{2}$ the shared state must be equivalent to the bipartite maximally entangled state $\ket{GHZ_2}$ but these inequalities do not satisfy the Mayers-Yao self-testing criterion   .

Finally, we have the equivalence class of  \textit{CHSH-like inequalities} of the form, 
\begin{eqnarray}
&\frac{1}{2}\Big(& \mean{A^{(1)}A^{(2)}A^{(3)}}+\mean{A^{(1)}A'^{(2)}A^{(3)}}\nonumber\\
& &+\mean{A'^{(1)}A^{(2)}A^{(3)}}-\mean{A'^{(1)}A'^{(2)}A^{(3)}}\Big)\leq 1.\nonumber
\end{eqnarray}
These inequalities are equivalent to the CHSH or $\mathcal{M}_2$ inequality for which the proof of Theorem \ref{MABKTheo} directly applies. The nonzero elements of the anti-diagonal matrix representation of the corresponding Bell operator are of the form,
 \begin{eqnarray}
 (1+e^{\iota k_1\theta_{1}}+e^{\iota k_2\theta_{2}}-e^{\iota (k_1\theta_{1}+k_2\theta_{2}}).
 \end{eqnarray}
The modulo of these values is $2\sqrt{1\pm\sin\theta_{1}\sin\theta_{2}}$, which implies that for quantum violation $\sqrt{2}$ the observables for the first pair of parties must maximally anti-commute, i.e., $\pm \theta_{1}=\pm\theta_{2}=\frac{\pi}{2}$ and shared state must be equivalent to the bipartite maximally entangled state $\ket{GHZ_2}$.

Apart from the linear inequalities considered here, it is easy to see that our proof technique, which relies on the anti-diagonal matrix representation of the Bell operator, directly yields the self-testing statements for all two settings binary outcome linear (on correlators) multipartite Bell inequalities. Next, we demonstrate that our proof technique can be applied to obtain self-testing statements for a large class of non-linear multipartite Bell inequalities.     

\subsection{$N$ party Uffink's quadratic inequalities}
As the convex set of bi-seprable multipartite quantum correlations does not form a polytope, linear inequalities like Svetlichny inequalities do not form tight efficient witnesses of genuine multipartite entanglement. on the other hand, the non-linear inequalities such as the Uffink's family of $N\geq 3$ party quadratic (on correlators) inequalities \eqref{Uff} better capture the boundary of the quantum set of bi-separable correlations, and hence form better witnesses of genuine multipartite quantum non-locality. Relying on the simple observation that for any two real numbers $x_1,x_2\in\mathbb{R}$, $x_1^2+x_2^2=|x_1+\iota x_2|^2$, we linearize the Uffink's $N\geq 3$ party quadratic inequalities and obtain self-testing statements based on anti-diagonal matrix representation of an Hermitian operator.
\begin{theorem}\label{UFFt}
In order to achieve maximal quantum violation of a $N$ party Uffnick's inequality, $\mean{\mathcal{U}_N}=2^{N+1}$ , the parties must share a $N\geq 3$ qubit GHZ state $\ket{GHZ_N}=\frac{1}{\sqrt{2}}(\ket{0}^{\otimes N}+ e^{\iota \phi_{ N}}\ket{1}^{\otimes N})$ and perform maximally anti-commuting projective measurements $A^{(j)}=\sigma_x$ and $A'^{(j)}=\sigma_y$ (upto local auxiliary systems and local isometries). 
\end{theorem}
\begin{proof}
We start with a linearization of the Bell criterion. Note that as both $\mathcal{U}_{N}$ and $\mathcal{U}'_{N}$ operators are Hermitian, $\mean{\mathcal{U}_N}, \mean{\mathcal{U}'_N} \in \mathbb{R}$, we have the following identity, 
\begin{equation}
\mathcal{U}^{\mathcal{M}}_{N}=\mean{\mathcal{U}_N}^2+\mean{\mathcal{U}'_N}^2=|\mean{\mathcal{U}_N}+\iota\mean{\mathcal{U}'_N}|^2.
\end{equation}
Our problem is hence reduced to studying the modulo of the mean value of non-Hermitian operator, $|\mathcal{\tilde{U}}_N|$, 
\begin{eqnarray}\label{tildeU}
\mathcal{\tilde{U}}_N=&&(\mathcal{U}_N+\iota\mathcal{U}'_N)\nonumber\\
=&&\left(\frac{1-\iota}{2}\right)^{N-1}\bigotimes_{j=1}^{N}(A^{[j]}+\iota A'^{[j]}).
\end{eqnarray}
From lemma \ref{spectrum} and theorem \ref{3TM}, without loss of generality, the local observables of any party $j$ can be taken to be, 
\begin{eqnarray} \nonumber
    & A^{(j)}&=\sigma_x, \\
    & A'^{(j)}&=\cos{\theta_{j}}\sigma_x+\sin{\theta_{j}}\sigma_y,
\end{eqnarray}
acting on the effective two dimensional subspace. This parametrization implies that the uffink operator  $\mathcal{\tilde{U}}_{N}$ \eqref{tildeU} has the following anti-diagonal matrix representation,
\begin{eqnarray} \label{UFFMatrix}
 \mathcal{\tilde{U}}_{N}  =\text{adiag} \begin{pmatrix}
\Big((\frac{1-\iota}{2})^{N-1}\Pi_{j=1}^{N}(1+\iota e^{-\iota\theta_j}\Big)\\
\vdots\\
\Big((\frac{1-\iota}{2})^{N-1}\Pi_{j=1}^{N}(1-\iota e^{\iota\theta_j}\Big)
\end{pmatrix},
\end{eqnarray}
where $\text{adiag}$ represent a matrix with values on the anti-diagonal. It is easy to see for any combination, $\forall j\in\{1,\ldots,N\}:\theta_j=\pm\frac{\pi}{2}$, one of the anti-diagonal element attains the maximum absolute value of $2^\frac{N-1}{2}$ while the others vanish.  

Yet again, for simplicity \footnote{Other subspaces are fully equivalent, with the sings standing in front of phases being the only difference.}, we consider on the subspace spanned by $\{\ket{0}^{\otimes N}, \ket{1}^{\otimes N}\}$. On this subspace, the modulo of the complex expectation value of $\mathcal{\tilde{U}}_N$ with respect to the state $\ket{\psi_N}=\alpha\ket{0}^{\otimes N}+ \beta\ket{1}^{\otimes N}$ has the expression,

\begin{eqnarray}
&|\bra{\psi_N}\tilde{\mathcal{U}}_{N}\ket{\psi_N}| & = \frac{1}{2^{\frac{N-1}{2}}}|\overline{\alpha}\Pi_{j=1}^{N}(1+\iota e^{-\iota\theta_j})\beta \nonumber \\
& & \hspace{25pt}+\overline{\beta}\Pi_{j=1}^{N}(1-\iota e^{\iota \theta_j})\alpha| \nonumber \\ 
& &\leq  \frac{1}{2^{\frac{N-1}{2}}}|\alpha||\beta|\Big(\Pi_{j=1}^{N}|(1+\iota e^{-\iota\theta_j})| \nonumber\\ 
& & \hspace{25pt}+ \Pi_{j=1}^{N}|(1-\iota e^{\iota \theta_j})|\Big). 
\end{eqnarray}
As $|\alpha|^2+|\beta|^2=1$ the maximum value of the above expression can only be attained for $|\alpha|=|\beta|=\frac{1}{\sqrt{2}}$, which picks out $\ket{GHZ_N}$ as the shared state. Consequently, we retrieve the following upper bound, 
\begin{eqnarray}
&|\bra{\psi_N}\tilde{\mathcal{U}}_{N}\ket{\psi_N}| & \leq \frac{1}{2^{\frac{N+1}{2}}}\Big(\Pi_{j=1}^{N}|(1+\iota  e^{-\iota\theta_j})| \nonumber \\
& & \left.\hspace{25pt}+\Pi_{j=1}^{N}|(1-\iota e^{\iota \theta_j})|\right) \nonumber \\
& & = 2^{\frac{N-1}{2}}\left(\Pi_{j=1}^{N}\left|\cos{\left(\frac{\pi}{4}+\frac{\theta_j}{2}\right)}\right|\right.  \nonumber\\
& & \left.\hspace{25pt}+\Pi_{j=1}^{N}\left|\sin{\left(\frac{\pi}{4}+\frac{\theta_j}{2}\right)}\right|\right) \nonumber.
\end{eqnarray}
Now, dropping the positive terms corresponding to any $N-2$ parties we obtain,
\begin{eqnarray}
& &\hspace{30pt}|\bra{\psi_N}\tilde{\mathcal{U}}_{N}\ket{\psi_N}|\nonumber\\
& & \leq  \nonumber {2^{\frac{N-1}{2}}}\Big(\left|\sin\left(\frac{2\theta_{i}+\pi}{4}\right)\sin\left(\frac{2\theta_{j}+\pi}{4}\right)\right| \\ 
&   & \hspace{21pt} + \left|\cos\left(\frac{2\theta_{i}+\pi}{4}\right)\cos\left(\frac{2\theta_{j}+\pi}{4}\right)\right|\Big)\nonumber \\ 
&  & =  2^{\frac{N-1}{2}} \max\Big\{\Big|\cos(\frac{\theta_{i}+\theta_{j}+\pi}{2})\Big|,\nonumber\\
& &  \hspace{30pt} \Big|\cos(\frac{\theta_{i}-\theta_{j}}{2}n)\Big|\Big\} \nonumber \\
& & \leq  2^{\frac{N-1}{2}},
\end{eqnarray}
Now, as the choice of $i,j$ is completely arbitrary the third inequality can only be saturated when $\forall j\in\{1,\ldots,N\}:\theta_j=\pm\frac{\pi}{2}$, i.e., for each party the local observables maximally anti-commute. Due to the quadratic form of the Uffink's inequalities, the state that maximally violates them is not unique, but it must be pure and LOCC-equivalent of $\ket{GHZ_N}$, i.e, the state is allowed to orbit under the action of local unitary groups. 

It remains to discuss the role of the junk state. Imagine that we have different realizations $\{R_j\}_j$ of the scenario such that they all give $\mean{\tilde{\mathcal{U}}_{N}}_{R_j}={2^{\frac{N-1}{2}}}e^{i\varphi_j}$. Realizations may differ in the orientation of the $\ket{GHZ_N}$ state and observables for each party, i.e., without loss of generality, the global state can always be written as $\sum_j\gamma_j\ket{R_j}\ket{GHZ_N}$, with $\{\ket{R_j}\}$ being orthogonal states of the rest of the system. The follows from the necessity that the observables are adjusted to the realization, requires that each orientation $\ket{R_j}$ is perfectly distinguishable from any other orientation $\ket{R_{j'}}$ by any individual party. In turn, this implies that by distinguishing between them, the parties can conditionally and locally transform each realization to the same form, ending up in the state $\left(\sum_j\gamma_j\ket{R_j}\right)\ket{GHZ_N}$.
\end{proof}

 \section{Conclusions}
 \label{CONC}
 Quantum correlations that violate Bell inequalities can certify the shared entangled state and local measurements in an entirely device-independent manner. This feature of quantum non-local correlations is referred to as self-testing. In this work, we presented a simple and broadly applicable proof technique to obtain self-testing statements for the maximum quantum violation of a large relevant class of multipartite  (involving arbitrary number $N$ of spatially separated parties) Bell inequalities. Unlike the relatively straightforward bipartite Bell scenarios, the multipartite scenarios are substantially richer in complexity owing to the various kinds of multipartite non-locality. 
 
 In contrast to the traditionally employed sum-of-squares-like proof techniques, which rely on the specific structure of the Bell inequality, the proof technique presented in this work applies to any Bell inequality (independent of its specific form) as it only relies on the anti-diagonal matrix representation of the Bell operator. In section \ref{maths} we show that the observables of each party in two setting binary outcome multipartite Bell scenarios can always be simultaneously represented as anti-diagonal matrices (theorem \ref{3TM}). Consequently, in such scenarios, the Hermitian Bell operators corresponding to all linear (on correlators) Bell inequalities also have an anti-diagonal matrix representation. In the section \ref{ST3W}, to demonstrate our proof technique, we obtain proofs of self-testing statements for the MABK family of $N$ party inequalities (theorem \ref{MABKTheo}) \footnote{This constitutes a reproduction of the self-testing statements for MABK inequalities from Ref. \cite{PhysRevA.95.062323}, and hence serves as a preliminary certification of our proof technique.}, followed by self-testing statements for tripartite  WWW\.ZB inequalities (section \ref{WWWZB1}). These inequalities exemplify all possible exceptions to perfect self-testing statements such as non-unique optimal states or observables, and non-anticommuting pairs of optimal measurements. To further demonstrate the versatility of our proof technique which relies only on the anti-diagonal matrix representation of the Bell operator and not even on its Hermiticity, we obtain self-testing statements for Uffink's family of quadratic Bell inequalities (theorem \ref{UFFt}). These quadratic Bell inequalities do not allow for sum-of-squares-like techniques which rely on the linearity of the Bell inequalities, and hence form the distinguishing application of our proof technique. 
 
 As self-testing statements completely certify the local sub-systems of each party, they can readily be used to certify the parties' private randomness, which in turn fuels cryptographic applications like device-independent randomness certification and conference key distribution schemes. As Uffink's quadratic inequalities form tighter witnesses of genuine multipartite non-locality compared to the linear inequalities, the corresponding self-testing statements can fuel information processing and communication tasks that require the participation of all parties, such as quantum secret sharing.       
 
\section{Acknowledgements}
We would like to thank Anubhav Chaturvedi for insightful discussions. This work is a part of NCN
Grant  No. 2017/26/E/ST2/01008. MW acknowledges partial support by the Foundation for Polish
Science (IRAP project, ICTQT, Contract No. 2018/MAB/5,
co-financed by EU within Smart Growth Operational Programme).
\section{Note}
During the preparation of the current work, we have become aware of \cite{Augusiak}. While it presents similar results, our proof technique, which constitutes the main conceptual contribution of this work, is substantially different from their proof technique. In particular, the key contrasting feature of our proof technique is its applicability to non-linear Bell inequalities.

\newpage
\bibliography{cite}
\end{document}